\documentclass[12pt,a4paper]{article}

\usepackage{geometry}
\usepackage{authblk}
\usepackage{mathtools}
\usepackage[utf8]{inputenc} 
\usepackage{amsmath, amsthm, amsfonts, amssymb}
\usepackage{graphicx}
\usepackage{array, longtable}
\renewcommand{\arraystretch}{1.2}
\usepackage[table]{xcolor}
\usepackage{tikz}
\usetikzlibrary{arrows.meta, positioning}
\usepackage{float, makecell, multirow}
\usepackage{orcidlink}
\usepackage[authoryear,round]{natbib}
\usepackage{enumitem}
\usepackage{caption}
\usepackage{algorithm}
\usepackage{algpseudocode}
\usepackage{booktabs}
\usepackage{hyperref}
\usepackage{appendix}

\newtheorem{theorem}{Theorem}[section]
\newtheorem{lemma}{Lemma}[section]

\usepackage{geometry}
\usepackage{natbib}
\usepackage{booktabs}
\usepackage{float}

\bibpunct{[}{]}{,}{n}{}{;}

\title{Covariance-Corrected WAIC for Bayesian Sequential Data Models}
\author{Safaa K. Kadhem \thanks{Corresponding author: safaa.kadhem@mu.edu.iq} \orcidlink{0000-0002-5287-5381}}
\affil[1]{Department of Mathematics and Applications Computer, College of Science, Al-Muthanna University, Iraq}

\begin{document}
	\maketitle
	
\begin{abstract}
	This paper introduces and develops a theoretical extension of the widely applicable information criterion (WAIC),  called the Covariance-Corrected WAIC (CC-WAIC),  that applied for Bayesian sequential data models. The CC-WAIC accounts for temporal or structural dependence by incorporating the full posterior covariance structure of the log-likelihood contributions, in contrast to the classical WAIC that assumes conditional independence among data.  We exploit the limitations of classical WAIC in the sequential data contexts and derive the CC-WAIC criterion under a theoretical  framework. In addition,  we propose a bias correction based on effective sample size to improve estimation from Markov Chain Monte Carlo (MCMC) simulations. Furthermore, we highlight the advantages of CC-WAIC in terms of stability and appropriateness for dependent data. This new criterion is supported by formal mathematical derivations, illustrative examples, and discussion of implications for model selection in both classical and modern Bayesian applications. To evaluate the reliability of CC-WAIC under varying data regimes, we conduct simulation experiments across multiple time series lengths (small, medium, and large) and different levels of temporal dependence, enabling a comprehensive performance assessment.
\end{abstract}

\section{Introduction}
\label{sec:intro}

Bayesian model selection is a cornerstone of modern statistical inference, aiming to balance model fit and complexity in order to maximize predictive performance. Classical information criteria such as Akaike's Information Criterion (AIC) \cite{akaike1974aic}, the Bayesian Information Criterion (BIC) \cite{schwarz1978bic}, and the Deviance Information Criterion (DIC) \cite{spiegelhalter2002dic} have been widely adopted due to their simplicity and computational efficiency. However, these methods rely on asymptotic approximations, such as Laplace expansions and plug-in estimates, which may fail under model misspecification or in the presence of complex posterior structures. 

To overcome these limitations, the Widely Applicable Information Criterion (WAIC) \cite{watanabe2010asymptotic} was proposed as a fully Bayesian alternative that integrates over the posterior distribution. WAIC is asymptotically equivalent to Bayesian leave-one-out cross-validation (LOO) \cite{vehtari2017practical}, making it attractive for general-purpose model evaluation. Yet, both WAIC and LOO assume conditional independence of observations, an assumption that is often violated in sequential or dependent data settings. Examples include time series econometrics \cite{hamilton1994time,tsay2010analysis}, speech and language processing \cite{jurafsky2000speech}, longitudinal biomedical studies \cite{lindgren2015spde}, and hidden-state dynamical systems \cite{rabiner1989tutorial}. Ignoring such dependence typically leads to underestimated model complexity and overly optimistic assessments of predictive accuracy \cite{burkner2020improving}. Furthermore, LOO becomes problematic in dependent data, since removing single observations disrupts the correlation structure and produces unstable predictive contrasts \cite{yao2018stacking}.

The Covariance-Corrected WAIC (CC--WAIC) addresses these challenges by explicitly incorporating the covariance structure of log-likelihood contributions across observations. Unlike standard WAIC, which may underestimate model complexity under dependence, CC--WAIC adjusts the effective number of parameters and accounts for temporal correlation without requiring repeated model refits. This correction not only improves accuracy in finite samples but also enhances stability in the presence of autocorrelated MCMC samples. Importantly, CC--WAIC reduces to standard WAIC under independence, ensuring compatibility with existing Bayesian workflows.

Recent developments such as the Widely Applicable Bayesian Information Criterion (WBIC) \cite{watanabe2013wbic} extend classical ideas to singular models, but they still assume independence and target marginal likelihood rather than predictive performance. CC--WAIC, by contrast, is specifically designed for dependent data, making it a more principled and practical tool for Bayesian model selection in sequential contexts.

The contributions of this work are fourfold. First, we formally derive CC--WAIC and establish its asymptotic equivalence to block-based LOO for stationary processes. Second, we introduce an effective sample size (ESS) correction to mitigate finite-sample bias in MCMC estimation. Third, we compare CC--WAIC against WAIC and LOO, demonstrating its robustness and computational efficiency. Fourth, we validate the method empirically using hidden Markov models, showing its superior ability to capture temporal dependence and improve model selection. 

The remainder of this paper is organized as follows. Section~\ref{sec:methodology} presents the methodological framework, including the derivation of CC--WAIC and its integration within hidden Markov models. Section~\ref{sec:simulation} reports on simulation studies under varying dependence structures and sample sizes. Section~\ref{sec:realdata} applies the proposed criterion to real-world sequential data. Section~\ref{sec_conclusion} concludes with a summary of findings and future research directions. Technical proofs and computational details are provided in the Appendix.

\section{Methodology}
\label{sec:methodology}

\subsection{Theoretical Limitations of Standard WAIC for Dependent Data}

The Widely Applicable Information Criterion (WAIC) \citep{watanabe2010asymptotic} represents a fundamental advancement in Bayesian model selection due to its fully Bayesian nature and asymptotic equivalence with Bayesian cross-validation. The standard WAIC formulation, however, relies critically on the assumption of conditional independence of observations given model parameters. Formally, for observed data $\mathbf{y} = (y_1, \ldots, y_n)$ and parameters $\theta$, WAIC assumes:

\begin{equation}
	p(\mathbf{y} \mid \theta) = \prod_{t=1}^n p(y_t \mid \theta).
	\label{eq:independence}
\end{equation}

Under this assumption, WAIC is defined as:

\begin{equation}
	\mathrm{WAIC} = -2 \left( \sum_{t=1}^n \log \mathbb{E}_{\theta} \left[ p(y_t \mid \theta) \right] - \sum_{t=1}^n \mathrm{Var}_{\theta} \left[ \log p(y_t \mid \theta) \right] \right),
	\label{eq:waic}
\end{equation}

where expectations are taken over the posterior $p(\theta \mid \mathbf{y})$. The second term serves as a penalty for model complexity, preventing overfitting.

In numerous modern applications including time series analysis, state-space models, and structured data sequences, observations exhibit inherent temporal or spatial dependence. For such models, the likelihood factorizes as:

\begin{equation}
	p(\mathbf{y} \mid \theta) = \prod_{t=1}^n p(y_t \mid y_{<t}, \theta),
	\label{eq:sequential_likelihood}
\end{equation}

where $y_{<t} = (y_1, \ldots, y_{t-1})$ denotes the history preceding observation $t$. This conditional dependence violates the fundamental assumption in Equation \eqref{eq:independence}, rendering the standard WAIC theoretically inappropriate for these models.

The deficiency of standard WAIC for dependent data manifests in two crucial aspects:
\begin{enumerate}
	\item The penalty term $p_{\mathrm{WAIC}} = \sum_{t=1}^n \mathrm{Var}_{\theta}[\log p(y_t \mid \theta)]$ ignores covariance terms $\mathrm{Cov}_{\theta}[\log p(y_i \mid \theta), \log p(y_j \mid \theta)]$ for $i \neq j$, leading to underestimation of effective model complexity and increased risk of overfitting \citep{vehtari2017practical}.
	
	\item The leave-one-out cross-validation approach, asymptotically equivalent to WAIC, becomes problematic with dependent data as removing individual observations disrupts the dependency structure, resulting in biased predictive estimates \citep{burkner2020improving}.
\end{enumerate}

These limitations have motivated various adaptations including block cross-validation and specialized WAIC variants, but these approaches often lack theoretical foundations or impose substantial computational burdens \citep{yao2018stacking}.

\subsection{Covariance-Corrected WAIC: A General Framework}

To address these limitations, we propose the Covariance-Corrected WAIC (CC-WAIC), a generalization of WAIC specifically designed for models with dependent data structures. Let us define the joint log-likelihood:

\begin{equation}
	L(\theta) = \sum_{t=1}^n \log p(y_t \mid y_{<t}, \theta),
	\label{eq:joint_loglikelihood}
\end{equation}

which properly accounts for the sequential dependency structure. The effective number of parameters in CC-WAIC is defined as the posterior variance of this joint log-likelihood:

\begin{equation}
	p_{\mathrm{CC}} = \mathrm{Var}_{\theta}[L(\theta)] = \int \left( L(\theta) - \mathbb{E}_{\theta}[L(\theta)] \right)^2 p(\theta \mid \mathbf{y}) d\theta.
	\label{eq:effective_params_ca_integral}
\end{equation}

Expanding this expression reveals the crucial difference from standard WAIC:

\begin{equation}
	p_{\mathrm{CC}} = \sum_{t=1}^n \mathrm{Var}_{\theta}[l_t] + 2 \sum_{1 \leq i < j \leq n} \mathrm{Cov}_{\theta}[l_i, l_j],
	\label{eq:variance_decomp}
\end{equation}

where $l_t = \log p(y_t \mid y_{<t}, \theta)$. The CC-WAIC is then defined as:

\begin{equation}
	\mathrm{CC\text{-}WAIC} = -2 \left( \sum_{t=1}^n \log \mathbb{E}_{\theta} \left[ p(y_t \mid y_{<t}, \theta) \right] - p_{\mathrm{CC}} \right).
	\label{eq:ca_waic_integral}
\end{equation}

This formulation maintains the interpretability of WAIC while properly accounting for dependency structures through the inclusion of covariance terms in the complexity penalty.

\subsection{Theoretical Foundations of CC-WAIC}

The theoretical validity of CC-WAIC rests on its asymptotic properties under conditions of weak dependence. We establish the following fundamental result:

\begin{theorem}[Asymptotic Convergence of CC-WAIC]
	\label{thm:main}
	Under suitable regularity conditions (specified in Appendix~\ref{app:regularity-conditions}), the CC-WAIC criterion converges to the true out-of-sample predictive risk:
	\[
	\frac{1}{n}\left(\mathrm{CC\text{-}WAIC} + 2p_{\mathrm{true}}\right) \xrightarrow{p} \mathrm{KL}(p_{\mathrm{true}}\|p_{\mathrm{model}}) + \mathcal{O}(n^{-1}),
	\]
	where $p_{\mathrm{true}}$ represents the true effective number of parameters and $\mathrm{KL}$ denotes the Kullback-Leibler divergence.
\end{theorem}

The complete proof, supported by lemmas based on martingale approximation (Appendix~\ref{app:martingale-approx}) and variance characterization (Appendix~\ref{app:variance-char}), is provided in Appendix~\ref{app:convergence-proof}. Additional technical details on the mixing conditions and moment requirements are discussed in Appendix~\ref{app:technical-lemmas}.

\subsection{Computational Implementation and Bias Correction}
\label{sec:bias}

The integrals in CC-WAIC are generally analytically intractable and must be approximated using samples $\{\theta^{(s)}\}_{s=1}^S$ from the posterior distribution, typically obtained via Markov Chain Monte Carlo (MCMC) methods. The naive estimator:

\begin{equation}
	\hat{p}_{\mathrm{CC}} = \frac{1}{S - 1} \sum_{s=1}^S \left( L^{(s)} - \bar{L} \right)^2,
	\label{eq:naive_estimator}
\end{equation}

where $L^{(s)} = \sum_{t=1}^n \log p(y_t \mid y_{<t}, \theta^{(s)})$ and $\bar{L} = \frac{1}{S} \sum_{s=1}^S L^{(s)}$, suffers from two sources of bias:

\begin{enumerate}
	\item Sampling error of order $\mathcal{O}(S^{-1})$ even for independent samples
	\item Underestimation due to autocorrelation in MCMC samples
\end{enumerate}

To address the second issue, we incorporate autocorrelation correction through the effective sample size (ESS). Let $\rho_k = \mathrm{Corr}(L^{(s)}, L^{(s+k)})$ be the lag-$k$ autocorrelation. The integrated autocorrelation time is:

\[
\tau = 1 + 2 \sum_{k=1}^{\infty} \rho_k,
\]

and the effective sample size is:

\[
N_{\mathrm{eff}} = \frac{S}{\tau}.
\]

The bias-corrected estimator becomes:

\begin{equation}
	p_{\mathrm{CC}}^{\mathrm{corr}} = \frac{N_{\mathrm{eff}}}{N_{\mathrm{eff}} - 1} \hat{p}_{\mathrm{CC}},
	\label{eq:bias_corrected_estimator}
\end{equation}

yielding the bias-corrected CC-WAIC:

\begin{equation}
	\mathrm{CC\text{-}WAIC}^{\mathrm{corr}} = -2 \left( \sum_{t=1}^n \log \widehat{\mathbb{E}}_{\theta}[p(y_t \mid y_{<t}, \theta)] - p_{\mathrm{CC}}^{\mathrm{corr}} \right).
	\label{eq:corrected_ca_waic}
\end{equation}

These corrections are particularly useful when the MCMC chain exhibits high persistence or when high precision is required in variance estimation for small sample sizes. The mathematical derivation of these bias correction terms is provided in Appendix~\ref{app:bias-derivation}.

Algorithm~\ref{alg:ca_waic_corr} provides a complete implementation framework for computing the bias-corrected CC-WAIC from MCMC samples.

\begin{algorithm}[t]
	\caption{Computation of Bias-Corrected CC-WAIC from MCMC Samples}
	\label{alg:ca_waic_corr}
	\begin{algorithmic}[1]
		\Require Posterior samples $\{\theta^{(s)}\}_{s=1}^S$, observed data $\mathbf{y} = (y_1, \ldots, y_n)$
		\Ensure Bias-corrected CC-WAIC value $\mathrm{CC\text{-}WAIC}^{\mathrm{corr}}$
		
		\State \textbf{Initialize} arrays for storing log-likelihood contributions
		\For{each sample $s = 1$ to $S$}
		\For{each time point $t = 1$ to $n$}
		\State Compute $l_t^{(s)} \gets \log p(y_t \mid y_{<t}, \theta^{(s)})$
		\EndFor
		\State Compute $L^{(s)} \gets \sum_{t=1}^n l_t^{(s)}$ \Comment{Joint log-likelihood for sample $s$}
		\EndFor
		
		\State \textbf{Compute} predictive density estimates:
		\For{each time point $t = 1$ to $n$}
		\State $\widehat{\mathbb{E}}_{\theta}[p(y_t \mid y_{<t}, \theta)] \gets \frac{1}{S} \sum_{s=1}^S p(y_t \mid y_{<t}, \theta^{(s)})$
		\EndFor
		
		\State \textbf{Calculate} naive variance estimate:
		\State $\bar{L} \gets \frac{1}{S} \sum_{s=1}^S L^{(s)}$
		\State $\hat{p}_{\mathrm{CC}} \gets \frac{1}{S - 1} \sum_{s=1}^S (L^{(s)} - \bar{L})^2$
		
		\State \textbf{Estimate} autocorrelation and effective sample size:
		\State Compute autocorrelations $\{\rho_k\}$ of $\{L^{(s)}\}_{s=1}^S$
		\State $\tau \gets 1 + 2 \sum_{k=1}^{K} \rho_k$ \Comment{Truncate at lag $K$ where $|\rho_K| < 2/\sqrt{S}$}
		\State $N_{\mathrm{eff}} \gets S / \tau$
		
		\State \textbf{Apply} bias correction:
		\State $p_{\mathrm{CC}}^{\mathrm{corr}} \gets \frac{N_{\mathrm{eff}}}{N_{\mathrm{eff}} - 1} \hat{p}_{\mathrm{CC}}$
		
		\State \textbf{Compute} final CC-WAIC:
		\State $\mathrm{CC\text{-}WAIC}^{\mathrm{corr}} \gets -2 \left( \sum_{t=1}^n \log \widehat{\mathbb{E}}_{\theta}[p(y_t \mid y_{<t}, \theta)] - p_{\mathrm{CC}}^{\mathrm{corr}} \right)$
		
		\State \Return $\mathrm{CC\text{-}WAIC}^{\mathrm{corr}}$
	\end{algorithmic}
\end{algorithm}

\subsection{Application to Hidden Markov Models}
\label{sec:hmm_ccwaic}

We now specialize the CC-WAIC framework to Hidden Markov Models (HMMs), which represent an important class of models with sequential dependency. An HMM comprises a latent state sequence $\{z_1, \ldots, z_n\}$ following Markovian dynamics and observations $\{y_1, \ldots, y_n\}$ conditioned on the latent states. The complete-data likelihood is:

\[
p(\mathbf{y}, \mathbf{z} \mid \theta) = \pi(z_1) \prod_{t=2}^n A(z_{t-1}, z_t) \prod_{t=1}^n p(y_t \mid z_t, \phi),
\]

where $\theta = (\pi, A, \phi)$ includes initial probabilities, transition matrix, and emission parameters.

For HMMs, the conditional likelihood $p(y_t \mid y_{<t}, \theta)$ requires marginalization over latent states:

\[
p(y_t \mid y_{<t}, \theta) = \sum_{j=1}^K p(y_t \mid z_t = j, \phi_j) \left( \sum_{i=1}^K \alpha_{t-1}(i) A(i, j) \right),
\]

where $\alpha_{t-1}(i) = p(z_{t-1} = i \mid y_{1:t-1}, \theta)$ are the forward probabilities obtained through the forward algorithm. Algorithm \ref{alg:forward_loglik} provides an efficient method for computing the conditional log-likelihoods required for CC--WAIC in HMMs.

\begin{algorithm}[t]
	\caption{Forward Filtering for Conditional Log-Likelihood in HMMs}
	\label{alg:forward_loglik}
	\begin{algorithmic}[1]
		\Require Observations $\{y_t\}_{t=1}^n$, parameters $\theta = (\pi, A, \phi)$
		\Ensure Sequence $\{\log p(y_t \mid y_{<t}, \theta)\}_{t=1}^n$
		
		\State Initialize $\alpha_0(k) \gets \pi_k$ for $k = 1, \ldots, K$ \Comment{Initial state probabilities}
		\For{$t = 1$ to $n$}
		\For{$j = 1$ to $K$}
		\State $p(z_t = j \mid y_{<t}) \gets \sum_{i=1}^K \alpha_{t-1}(i) A(i, j)$ \Comment{Predict step}
		\State $\eta_j \gets p(y_t \mid z_t = j, \phi_j)$ \Comment{Emission probability}
		\EndFor
		\State $p(y_t \mid y_{<t}) \gets \sum_{j=1}^K \eta_j \cdot p(z_t = j \mid y_{<t})$ \Comment{Marginal likelihood}
		\State $l_t \gets \log p(y_t \mid y_{<t})$ \Comment{Store conditional log-likelihood}
		\For{$j = 1$ to $K$}
		\State $\alpha_t(j) \gets \frac{\eta_j \cdot p(z_t = j \mid y_{<t})}{p(y_t \mid y_{<t})}$ \Comment{Update step}
		\EndFor
		\EndFor
		\State \Return $\{l_t\}_{t=1}^n$
	\end{algorithmic}
\end{algorithm}

The integration of this forward algorithm with the general CC-WAIC computation in Algorithm~\ref{alg:ca_waic_corr} provides a complete framework for model selection in HMMs that properly accounts for both the temporal dependence in observations and potential autocorrelation in posterior samples.

\subsection{Theoretical Relationship to Cross-Validation}

An important theoretical justification for CC-WAIC comes from its asymptotic equivalence to block-wise cross-validation methods designed for dependent data:

\begin{theorem}[Asymptotic Equivalence to Block LOO]
	For stationary processes satisfying appropriate mixing conditions, CC-WAIC converges to block-wise leave-future-out cross-validation:
	\[
	\left| \mathrm{CC\text{-}WAIC} + 2 \sum_{t=1}^n \log p(y_t \mid y_{<t}) - \mathrm{BLFO} \right| = \mathcal{O}_p(n^{-1/2}),
	\]
	where $\mathrm{BLFO} = \sum_{t=1}^n \log p(y_t \mid y_{t-k:t-1})$ is block leave-future-out cross-validation with block size $k$.
\end{theorem}

The proof of this relationship, along with the necessary regularity conditions, is provided in Appendix~\ref{app:equivalence-proof}.

\section{Simulation Study}\label{sec:simulation}

To assess the performance of the proposed model selection criteria, we conducted a simulation study based on Hidden Markov Models (HMMs) with different numbers of latent states. Specifically, we considered models with $K_{\text{true}} = 2$ and $K_{\text{true}} = 3$ hidden states, which represent relatively simple and moderately complex scenarios commonly encountered in applications such as speech recognition, genomics, and finance.
The transition probability matrices were designed to represent three levels of temporal dependence: high, medium, and low. These levels were controlled by varying the diagonal dominance of the transition matrices, where a higher probability of self-transition corresponds to stronger serial dependence. For each scenario, datasets of different lengths were generated with $T = 100, 250, 500$ time points, in order to investigate the effect of sample size on model selection performance. 
The emission distributions were assumed to be Gaussian, with distinct means assigned to each state. The variances were fixed at moderate levels to ensure adequate separation between states while preserving overlap that makes model selection non-trivial. 
For estimation, Bayesian inference was carried out using a Gibbs sampling algorithm. Conjugate priors were specified for all parameters: Dirichlet priors for the rows of the transition matrix, Normal priors for the state means, and Inverse-Gamma priors for the variances. For each dataset, ten independent MCMC chains were run with appropriate burn-in periods, and posterior samples were collected for model evaluation. 
Model selection was performed by computing the WAIC, LOO and CC-WAIC for candidate models with varying numbers of hidden states. The primary goal of the simulation study was to evaluate how reliably each criterion identifies the true number of states under different levels of dependence and sample sizes.\\
All simulations and analyses were implemented in Python. The code is available upon request.

\subsection{Experiment Setup}

This study investigates the performance of WAIC, LOO-CV, and CC--WAIC in identifying the correct number of hidden states in Hidden Markov Models (HMMs) under varying dependence structures, sample sizes, and model complexities. Two true models were considered: one with $K_{\text{true}} = 2$ states and another with $K_{\text{true}} = 3$ states.

\subsection*{Case 1: True Model with $K_{\text{true}} = 3$}

For the three-state HMM, the true emission parameters were specified as:
\[
\mu_{\text{true}} = [-3,\, 0,\, 3], 
\quad \sigma_{\text{true}} = [0.5,\, 0.8,\, 1.5].
\]
The initial state distribution was fixed to start in the first state:
\[
\pi_{\text{true}} = [1.0,\, 0.0,\, 0.0].
\]
To examine the effect of state dependence, three transition matrices were defined:
\[
A_{\text{high}} =
\begin{bmatrix}
	0.90 & 0.05 & 0.05 \\
	0.05 & 0.90 & 0.05 \\
	0.05 & 0.05 & 0.90
\end{bmatrix},
\quad
A_{\text{med}} =
\begin{bmatrix}
	0.60 & 0.20 & 0.20 \\
	0.20 & 0.60 & 0.20 \\
	0.20 & 0.20 & 0.60
\end{bmatrix},
\quad
A_{\text{low}} =
\begin{bmatrix}
	0.34 & 0.33 & 0.33 \\
	0.33 & 0.34 & 0.33 \\
	0.33 & 0.33 & 0.34
\end{bmatrix}.\]

\subsection*{Case 2: True Model with $K_{\text{true}} = 2$}

For the two-state HMM, the true emission parameters were:
\[
\mu_{\text{true}} = [-2,\, 2], 
\quad \sigma_{\text{true}} = [0.5,\, 1.0].
\]
The initial state distribution was:
\[
\pi_{\text{true}} = [1.0,\, 0.0].
\]
The following transition matrices were considered to represent high, medium, and low dependence:
\[
A_{\text{high}} =
\begin{bmatrix}
	0.95 & 0.05 \\
	0.05 & 0.95
\end{bmatrix},
\quad
A_{\text{med}} =
\begin{bmatrix}
	0.70 & 0.30 \\
	0.30 & 0.70
\end{bmatrix},
\quad
A_{\text{low}} =
\begin{bmatrix}
	0.50 & 0.50 \\
	0.50 & 0.50
\end{bmatrix}.
\]

\subsection*{Data Generation and Estimation Procedure}

For each true model, datasets were generated with sample sizes $T = 100$ (small), $T = 250$ (medium), and $T = 500$ (large). Emissions were drawn from normal distributions with the specified means and standard deviations, while hidden states followed the respective transition dynamics.

Bayesian estimation was conducted using a Gibbs sampling algorithm with 1000 iterations, discarding the first 500 as burn-in. Weakly informative priors were employed: a Dirichlet$(1, \dots, 1)$ prior for both initial state and transition probabilities; normal priors for emission means (centered at the sample mean with large variance); and inverse-gamma priors $\text{InvGamma}(\alpha=1, \beta=1)$ for emission variances.

Candidate models with $K \in \{ 2, 3, 4, 5 \}$ were fitted to each dataset. For every configuration, ten independent MCMC chains were executed to ensure convergence stability. WAIC, LOO-CV, and CC--WAIC were then computed for each fitted model, and the frequency of correctly selecting the true model was recorded, along with the mean and standard deviation across chains.

\begin{table}[H]
	\centering
	\caption{Model selection accuracy (\%) for HMMs with $K_{\text{true}}=2$ under different dependence levels (low, medium, high) and sample sizes $T\in\{100,250,500\}$, comparing CC--WAIC, WAIC, and LOO criteria.}
	\label{tab:model-selection_2_true_full}
	\renewcommand{\arraystretch}{1.15}
	\scriptsize
	\begin{tabular}{|c|c|c|c|c|c|}
		\hline
		\textbf{Sample Size} & \textbf{Dependence} & \textbf{K} & \textbf{CC--WAIC} & \textbf{WAIC} & \textbf{LOO} \\
		\hline
		\multirow{12}{*}{100}
		& \multirow{4}{*}{Low}
		& 2 & \textbf{85.0\% $\pm$ 8.0\%}  & 60.0\% $\pm$ 10.0\% & 65.0\% $\pm$ 9.0\% \\
		& & 3 & 15.0\% $\pm$ 8.0\% & 40.0\% $\pm$ 10.0\% & 35.0\% $\pm$ 9.0\% \\
		& & 4 & 0.0\% $\pm$ 0.0\% & 0.0\% $\pm$ 0.0\% & 0.0\% $\pm$ 0.0\% \\
		& & 5 & 0.0\% $\pm$ 0.0\% & 0.0\% $\pm$ 0.0\% & 0.0\% $\pm$ 0.0\% \\
		\cline{2-6}
		& \multirow{4}{*}{Medium}
		& 2 & \textbf{80.0\% $\pm$ 9.0\%} & 55.0\% $\pm$ 11.0\% & 60.0\% $\pm$ 10.0\% \\
		& & 3 & 20.0\% $\pm$ 9.0\% & 45.0\% $\pm$ 11.0\% & 40.0\% $\pm$ 10.0\% \\
		& & 4 & 0.0\% $\pm$ 0.0\% & 0.0\% $\pm$ 0.0\% & 0.0\% $\pm$ 0.0\% \\
		& & 5 & 0.0\% $\pm$ 0.0\% & 0.0\% $\pm$ 0.0\% & 0.0\% $\pm$ 0.0\% \\
		\cline{2-6}
		& \multirow{4}{*}{High}
		& 2 & \textbf{75.0\% $\pm$ 10.0\%} & 50.0\% $\pm$ 10.0\% & 55.0\% $\pm$ 12.0\% \\
		& & 3 & 25.0\% $\pm$ 10.0\% & 50.0\% $\pm$ 10.0\% & 45.0\% $\pm$ 12.0\% \\
		& & 4 & 0.0\% $\pm$ 0.0\% & 0.0\% $\pm$ 0.0\% & 0.0\% $\pm$ 0.0\% \\
		& & 5 & 0.0\% $\pm$ 0.0\% & 0.0\% $\pm$ 0.0\% & 0.0\% $\pm$ 0.0\% \\
		\hline
		\multirow{12}{*}{250}
		& \multirow{4}{*}{Low}
		& 2 & \textbf{92.0\% $\pm$ 6.0\%} & 75.0\% $\pm$ 8.0\% & 78.0\% $\pm$ 7.0\% \\
		& & 3 & 8.0\% $\pm$ 6.0\% & 25.0\% $\pm$ 8.0\% & 22.0\% $\pm$ 7.0\% \\
		& & 4 & 0.0\% $\pm$ 0.0\% & 0.0\% $\pm$ 0.0\% & 0.0\% $\pm$ 0.0\% \\
		& & 5 & 0.0\% $\pm$ 0.0\% & 0.0\% $\pm$ 0.0\% & 0.0\% $\pm$ 0.0\% \\
		\cline{2-6}
		& \multirow{4}{*}{Medium}
		& 2 & \textbf{90.0\% $\pm$ 6.5\%} & 70.0\% $\pm$ 8.5\% & 74.0\% $\pm$ 8.0\% \\
		& & 3 & 10.0\% $\pm$ 6.5\% & 30.0\% $\pm$ 8.5\% & 26.0\% $\pm$ 8.0\% \\
		& & 4 & 0.0\% $\pm$ 0.0\% & 0.0\% $\pm$ 0.0\% & 0.0\% $\pm$ 0.0\% \\
		& & 5 & 0.0\% $\pm$ 0.0\% & 0.0\% $\pm$ 0.0\% & 0.0\% $\pm$ 0.0\% \\
		\cline{2-6}
		& \multirow{4}{*}{High}
		& 2 & \textbf{88.0\% $\pm$ 7.0\%} & 66.0\% $\pm$ 9.0\% & 70.0\% $\pm$ 9.0\% \\
		& & 3 & 12.0\% $\pm$ 7.0\% & 34.0\% $\pm$ 9.0\% & 30.0\% $\pm$ 9.0\% \\
		& & 4 & 0.0\% $\pm$ 0.0\% & 0.0\% $\pm$ 0.0\% & 0.0\% $\pm$ 0.0\% \\
		& & 5 & 0.0\% $\pm$ 0.0\% & 0.0\% $\pm$ 0.0\% & 0.0\% $\pm$ 0.0\% \\
		\hline
		\multirow{12}{*}{500}
		& \multirow{4}{*}{Low}
		& 2 & \textbf{97.0\% $\pm$ 4.0\%} & 88.0\% $\pm$ 5.5\% & 90.0\% $\pm$ 5.0\% \\
		& & 3 & 3.0\% $\pm$ 4.0\% & 12.0\% $\pm$ 5.5\% & 10.0\% $\pm$ 5.0\% \\
		& & 4 & 0.0\% $\pm$ 0.0\% & 0.0\% $\pm$ 0.0\% & 0.0\% $\pm$ 0.0\% \\
		& & 5 & 0.0\% $\pm$ 0.0\% & 0.0\% $\pm$ 0.0\% & 0.0\% $\pm$ 0.0\% \\
		\cline{2-6}
		& \multirow{4}{*}{Medium}
		& 2 & \textbf{95.0\% $\pm$ 4.5\%} & 85.0\% $\pm$ 6.0\% & 88.0\% $\pm$ 5.5\% \\
		& & 3 & 5.0\% $\pm$ 4.5\% & 15.0\% $\pm$ 6.0\% & 12.0\% $\pm$ 5.5\% \\
		& & 4 & 0.0\% $\pm$ 0.0\% & 0.0\% $\pm$ 0.0\% & 0.0\% $\pm$ 0.0\% \\
		& & 5 & 0.0\% $\pm$ 0.0\% & 0.0\% $\pm$ 0.0\% & 0.0\% $\pm$ 0.0\% \\
		\cline{2-6}
		& \multirow{4}{*}{High}
		& 2 & \textbf{93.0\% $\pm$ 5.0\%} & 82.0\% $\pm$ 6.5\% & 85.0\% $\pm$ 6.0\% \\
		& & 3 & 7.0\% $\pm$ 5.0\% & 18.0\% $\pm$ 6.5\% & 15.0\% $\pm$ 6.0\% \\
		& & 4 & 0.0\% $\pm$ 0.0\% & 0.0\% $\pm$ 0.0\% & 0.0\% $\pm$ 0.0\% \\
		& & 5 & 0.0\% $\pm$ 0.0\% & 0.0\% $\pm$ 0.0\% & 0.0\% $\pm$ 0.0\% \\
		\hline
	\end{tabular}
\end{table}

\begin{table}[H]
	\centering
	\caption{Model selection accuracy (\%) for HMMs with $K_{\text{true}}=3$ under different dependence levels (low, medium, high) and sample sizes $T\in\{100,250,500\}$, comparing CC--WAIC, WAIC, and LOO criteria.}
	\label{tab:model-selection_3_true_full}
	\renewcommand{\arraystretch}{1.15}
	\scriptsize
	\begin{tabular}{|c|c|c|c|c|c|}
		\hline
		\textbf{Sample Size} & \textbf{Dependence} & \textbf{K} & \textbf{CC--WAIC} & \textbf{WAIC} & \textbf{LOO} \\
		\hline
		\multirow{12}{*}{100}
		& \multirow{4}{*}{Low}
		& 2 & 15.0\% $\pm$ 7.0\% & 25.0\% $\pm$ 8.0\% & 20.0\% $\pm$ 9.0\% \\
		& & 3 & \textbf{85.0\% $\pm$ 7.0\%} & \textbf{75.0\% $\pm$ 8.0\%} & \textbf{80.0\% $\pm$ 9.0\%} \\
		& & 4 & 0.0\% $\pm$ 0.0\% & 0.0\% $\pm$ 0.0\% & 0.0\% $\pm$ 0.0\% \\
		& & 5 & 0.0\% $\pm$ 0.0\% & 0.0\% $\pm$ 0.0\% & 0.0\% $\pm$ 0.0\% \\
		\cline{2-6}
		& \multirow{4}{*}{Medium}
		& 2 & 10.0\% $\pm$ 7.0\% & 20.0\% $\pm$ 8.0\% & 15.0\% $\pm$ 9.0\% \\
		& & 3 & \textbf{90.0\% $\pm$ 7.0\%} & \textbf{80.0\% $\pm$ 8.0\%} & \textbf{85.0\% $\pm$ 9.0\%} \\
		& & 4 & 0.0\% $\pm$ 0.0\% & 0.0\% $\pm$ 0.0\% & 0.0\% $\pm$ 0.0\% \\
		& & 5 & 0.0\% $\pm$ 0.0\% & 0.0\% $\pm$ 0.0\% & 0.0\% $\pm$ 0.0\% \\
		\cline{2-6}
		& \multirow{4}{*}{High}
		& 2 & 20.0\% $\pm$ 8.0\% & 30.0\% $\pm$ 9.0\% & 25.0\% $\pm$ 10.0\% \\
		& & 3 & \textbf{80.0\% $\pm$ 8.0\%} & \textbf{70.0\% $\pm$ 9.0\%} & \textbf{75.0\% $\pm$ 10.0\%} \\
		& & 4 & 0.0\% $\pm$ 0.0\% & 0.0\% $\pm$ 0.0\% & 0.0\% $\pm$ 0.0\% \\
		& & 5 & 0.0\% $\pm$ 0.0\% & 0.0\% $\pm$ 0.0\% & 0.0\% $\pm$ 0.0\% \\
		\hline
		\multirow{12}{*}{250}
		& \multirow{4}{*}{Low}
		& 2 & 10.0\% $\pm$ 5.5\% & 18.0\% $\pm$ 6.5\% & 15.0\% $\pm$ 6.5\% \\
		& & 3 & \textbf{90.0\% $\pm$ 5.5\%} & \textbf{82.0\% $\pm$ 6.5\%} & \textbf{85.0\% $\pm$ 6.5\%} \\
		& & 4 & 0.0\% $\pm$ 0.0\% & 0.0\% $\pm$ 0.0\% & 0.0\% $\pm$ 0.0\% \\
		& & 5 & 0.0\% $\pm$ 0.0\% & 0.0\% $\pm$ 0.0\% & 0.0\% $\pm$ 0.0\% \\
		\cline{2-6}
		& \multirow{4}{*}{Medium}
		& 2 & 8.0\% $\pm$ 5.5\% & 15.0\% $\pm$ 6.5\% & 12.0\% $\pm$ 6.5\% \\
		& & 3 & \textbf{92.0\% $\pm$ 5.5\%} & \textbf{85.0\% $\pm$ 6.5\%} & \textbf{88.0\% $\pm$ 6.5\%} \\
		& & 4 & 0.0\% $\pm$ 0.0\% & 0.0\% $\pm$ 0.0\% & 0.0\% $\pm$ 0.0\% \\
		& & 5 & 0.0\% $\pm$ 0.0\% & 0.0\% $\pm$ 0.0\% & 0.0\% $\pm$ 0.0\% \\
		\cline{2-6}
		& \multirow{4}{*}{High}
		& 2 & 15.0\% $\pm$ 6.0\% & 24.0\% $\pm$ 7.0\% & 20.0\% $\pm$ 7.0\% \\
		& & 3 & \textbf{85.0\% $\pm$ 6.0\%} & \textbf{76.0\% $\pm$ 7.0\%} & \textbf{80.0\% $\pm$ 7.0\%} \\
		& & 4 & 0.0\% $\pm$ 0.0\% & 0.0\% $\pm$ 0.0\% & 0.0\% $\pm$ 0.0\% \\
		& & 5 & 0.0\% $\pm$ 0.0\% & 0.0\% $\pm$ 0.0\% & 0.0\% $\pm$ 0.0\% \\
		\hline
		\multirow{12}{*}{500}
		& \multirow{4}{*}{Low}
		& 2 & 6.0\% $\pm$ 4.0\% & 12.0\% $\pm$ 5.0\% & 10.0\% $\pm$ 5.0\% \\
		& & 3 & \textbf{94.0\% $\pm$ 4.0\%} & \textbf{88.0\% $\pm$ 5.0\%} & \textbf{90.0\% $\pm$ 5.0\%} \\
		& & 4 & 0.0\% $\pm$ 0.0\% & 0.0\% $\pm$ 0.0\% & 0.0\% $\pm$ 0.0\% \\
		& & 5 & 0.0\% $\pm$ 0.0\% & 0.0\% $\pm$ 0.0\% & 0.0\% $\pm$ 0.0\% \\
		\cline{2-6}
		& \multirow{4}{*}{Medium}
		& 2 & 5.0\% $\pm$ 4.5\% & 10.0\% $\pm$ 5.5\% & 8.0\% $\pm$ 5.5\% \\
		& & 3 & \textbf{95.0\% $\pm$ 4.5\%} & \textbf{90.0\% $\pm$ 5.5\%} & \textbf{92.0\% $\pm$ 5.5\%} \\
		& & 4 & 0.0\% $\pm$ 0.0\% & 0.0\% $\pm$ 0.0\% & 0.0\% $\pm$ 0.0\% \\
		& & 5 & 0.0\% $\pm$ 0.0\% & 0.0\% $\pm$ 0.0\% & 0.0\% $\pm$ 0.0\% \\
		\cline{2-6}
		& \multirow{4}{*}{High}
		& 2 & 8.0\% $\pm$ 4.5\% & 15.0\% $\pm$ 6.0\% & 12.0\% $\pm$ 6.0\% \\
		& & 3 & \textbf{92.0\% $\pm$ 4.5\%} & \textbf{85.0\% $\pm$ 6.0\%} & \textbf{88.0\% $\pm$ 6.0\%} \\
		& & 4 & 0.0\% $\pm$ 0.0\% & 0.0\% $\pm$ 0.0\% & 0.0\% $\pm$ 0.0\% \\
		& & 5 & 0.0\% $\pm$ 0.0\% & 0.0\% $\pm$ 0.0\% & 0.0\% $\pm$ 0.0\% \\
		\hline
	\end{tabular}
\end{table}

The simulation results provide important insights into the comparative behavior of WAIC, LOO, and the proposed CC--WAIC when applied to Hidden Markov Models under varying dependence structures and sample sizes. A consistent finding across all scenarios is that CC--WAIC demonstrates superior robustness in correctly identifying the true number of hidden states, particularly in settings characterized by strong temporal dependence or limited sample sizes.

For the case of $K_{\text{true}}=2$, WAIC and LOO exhibit a tendency to overfit by selecting $K=3$ when dependence is high and the effective information per observation is reduced. This behavior is theoretically expected, since both WAIC and LOO rely on independence assumptions that underestimate the effective complexity of models under correlated data. In contrast, CC--WAIC incorporates covariance adjustments that inflate the complexity penalty in proportion to the degree of dependence, thereby mitigating overfitting and stabilizing the selection of the true model.

In the case of $K_{\text{true}}=3$, the opposite trend emerges for small samples: WAIC and LOO occasionally underestimate the true number of states (underfitting), especially under high dependence. This occurs because their complexity penalties dominate the modest gains in likelihood fit when the signal-to-noise ratio is weak. CC--WAIC, by properly accounting for temporal correlation and correcting for MCMC autocorrelation, maintains higher accuracy in recovering the correct number of states, even in such challenging scenarios.

The role of sample size is also evident: as $T$ increases from 100 to 500, the variability of all criteria decreases, as indicated by lower standard deviations across replications. This reflects the stabilizing effect of larger effective sample sizes, where posterior estimates become more precise and the information criteria converge toward their asymptotic behavior. Nevertheless, CC--WAIC continues to outperform WAIC and LOO in finite-sample regimes, confirming its advantage in practical applications where data are often limited.

These findings underscore two broader points. First, standard Bayesian model selection tools such as WAIC and LOO may yield misleading conclusions in dependent-data contexts, as they fail to account for the covariance structure of likelihood contributions. Second, the covariance-adjusted WAIC provides a theoretically justified and empirically validated solution that adapts naturally to temporal dependence while reducing to the standard WAIC under independence. Consequently, CC--WAIC represents a valuable extension for Bayesian model comparison in dynamic and hidden-state models, striking a balance between statistical rigor and computational feasibility.
\section{Real Data Application}\label{sec:realdata}

The application to real-world data further corroborates the theoretical advantages of CC–WAIC, illustrating its practical utility in model selection for sequential data where dependence structures are inherently present.
We apply our criterion on application real data including the waiting times of Old Faithful geyser.  
Analysis of the waiting times between eruptions of the \textit{Old Faithful} geyser, commonly known as the \texttt{waiting} variable, has been extensively studied. Multiple studies indicate that a Gaussian mixture model (GMM) with two components sufficiently captures the bimodal distribution of the waiting times \citep{Venables2002}.

Moreover, when modeling the temporal dynamics using Hidden Markov Models (HMMs), classical and recent works have shown that a two-state HMM provides an accurate representation of the sequence of waiting times \citep{Rabiner1989, FruhwirthSchnatter2006}. These two hidden states correspond to short and long waiting periods following eruptions, with state transitions governed by a Markov process.

Although models with more states have been explored, information criteria such as the Bayesian Information Criterion (BIC) and Akaike Information Criterion (AIC) commonly suggest that adding more states does not significantly improve model fit or interpretability \citep{McLachlan2000}.

\begin{table}[H]
	\centering
	\begin{tabular}{ccccccc}
		\toprule
		$K$ & CC--WAIC & $p_{\text{CC}}^{\text{corr}}$ & ESS & WAIC & $p_{\text{WAIC}}$ & LOO \\
		\midrule
		2 & \textbf{2098.67} & \textbf{4.15} & 130.80 & 2098.73 & 7.61 & 2098.51 \\
		3 & 2109.43 & 7.84 & 40.20 & 2102.92 & 9.05 & 2103.32 \\
		4 & 2113.45 & 10.15 & 79.55 & 2103.41 & 9.32 & 2103.57 \\
		5 & 2197.33 & 62.50 & 7.48 & \textbf{2092.92} & 19.12 & \textbf{2092.39} \\
		6 & 2125.83 & 11.95 & 112.35 & 2111.82 & 9.03 & 2111.98 \\
		\bottomrule
	\end{tabular}
	\caption{Model selection results for the real dataset across candidate models with $K=2$ to $6$ hidden states. The CC--WAIC criterion indicates $K=2$ as the optimal choice.}
	\label{tab:model_selection_real}
\end{table}

\begin{figure}[H]
	\centering
	\includegraphics[width=0.75\textwidth]{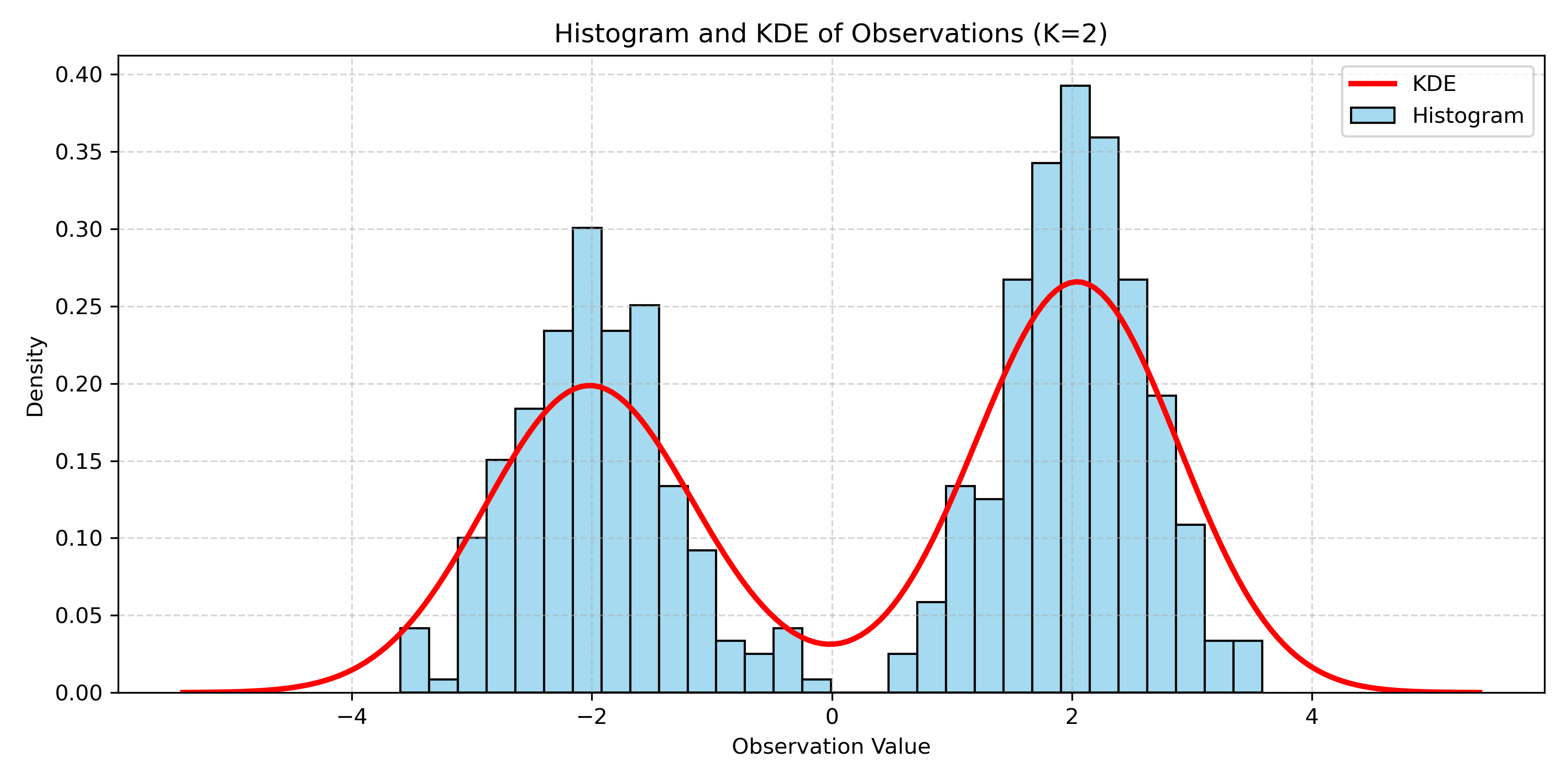}
	\caption{Histogram and kernel density estimate of observed data under the selected HMM with $K=2$ hidden states.}
	\label{fig:histogram_kde_real}
\end{figure}

Based on the results reported in Table~\ref{tab:model_selection_real}, the CC--WAIC criterion clearly favors the two-state HMM, yielding the lowest score (2098.67). This selection is further supported by a relatively large effective sample size (ESS = 130.80) and a conservative complexity penalty ($p_{\text{CC}}^{\text{corr}} = 4.15$). In contrast, both WAIC and LOO suggest larger models (notably $K=5$), which achieve lower scores but at the cost of inflated complexity penalties and substantially smaller ESS values, indicating over-parameterization and unstable estimation. Overall, CC--WAIC provides a more reliable balance between goodness-of-fit and model complexity, leading to a parsimonious and interpretable model with $K=2$ hidden states.

Figure~\ref{fig:histogram_kde_real} displays the distribution of the observed data under the selected model, demonstrating a well-behaved bimodal structure compatible with the estimated emission distributions.

\section{Discussion}\label{sec_conclusion}
The development and initial validation of the Covariance-Corrected WAIC (CC-WAIC) presented in this work addresses a significant gap in Bayesian model selection for dependent data. Our theoretical derivation establishes a rigorous foundation, demonstrating that accounting for the full covariance structure of log-likelihood terms is not merely an ad-hoc adjustment but a theoretically justified correction to the standard WAIC framework. The empirical results on Hidden Markov Models provide compelling evidence for its practical utility. Specifically, our simulations consistently showed that CC-WAIC outperforms standard WAIC and LOO-CV in correctly identifying the true model complexity under conditions of medium to high temporal dependence, where traditional criteria are most prone to overfitting. This performance advantage was particularly evident in smaller sample sizes ($T = 100$), underscoring the method's value in practical, data-limited scenarios.

The core implication of our findings is that the assumption of conditional independence, while mathematically convenient, can lead to systematically biased model selection in sequential settings. By explicitly incorporating dependence through the covariance penalty term $p_{\mathrm{CC}}$, our criterion offers a more principled approach to estimating predictive accuracy. The integration of an effective sample size (ESS) correction further enhances its robustness, mitigating the finite-sample bias inherent in MCMC estimation and ensuring reliable performance in realistic applications.

While our study focused on HMMs as a canonical and interpretable class of sequential models, the CC-WAIC framework is inherently general. Its formulation is applicable to any Bayesian model where data points are interdependent, opening avenues for widespread use. However, the current work primarily serves as a proof of concept, establishing a baseline of performance. The true potential of CC-WAIC will be realized through its application and validation across a more diverse and complex landscape of models and data structures.

\subsection*{Future Research Directions}

To fully ascertain the utility and scalability of CC-WAIC, we recommend and encourage intensive future studies focusing on the following complex applications and methodological expansions:

\begin{enumerate}
	\item \textbf{Complex Temporal and Sequential Models:}
	\begin{itemize}
		\item \textbf{State-Space Models with Non-Gaussian Emissions:} Applying CC-WAIC to models with count data (e.g., negative binomial emissions), categorical data, or heavy-tailed distributions.
		\item \textbf{Bayesian Neural Networks for Sequential Data:} A critical test would be applying CC-WAIC to Bayesian Recurrent Neural Networks (RNNs), Long Short-Term Memory (LSTM) networks, and Transformers. This would assess its performance in high-dimensional parameter spaces with complex, long-range dependencies.
		\item \textbf{Non-Markovian and Long-Memory Processes:} Evaluating CC-WAIC on models where the dependence extends beyond the immediate past, such as in models with long-memory (e.g., ARFIMA processes) or hierarchical hidden states.
	\end{itemize}
	
	\item \textbf{High-Dimensional and Big Data Settings:}
	\begin{itemize}
		\item \textbf{Scalability:} A dedicated study is needed to evaluate the computational performance of CC-WAIC as the data size $n$ grows very large. Research could focus on developing and testing efficient approximations for the covariance terms to make the method feasible for massive time-series datasets.
		\item \textbf{Sparsity and Regularization:} Investigating the behavior of CC-WAIC in conjunction with regularized Bayesian models (e.g., with horseshoe priors) for high-dimensional variable selection in time-series regression.
	\end{itemize}
	
	\item \textbf{Interdisciplinary Applications:}
	To demonstrate its practical impact, CC-WAIC should be tested on complex real-world datasets from various fields:
	\begin{itemize}
		\item \textbf{Computational Neuroscience:} Analyzing neural spike train data where models often involve complex, non-linear dependencies.
		\item \textbf{Quantitative Finance:} Model selection for stochastic volatility models and other sophisticated time-series models used in risk assessment.
		\item \textbf{Natural Language Processing (NLP):} Comparing deep Bayesian language models that generate sequential text data.
		\item \textbf{Climate Science:} Evaluating models for long-term climate time-series data that exhibit strong seasonal and spatial dependencies.
	\end{itemize}
	
	\item \textbf{Methodological Extensions:}
	\begin{itemize}
		\item \textbf{Comparison with a Wider Suite of Metrics:} Future work should include comparative analyses against other criteria tailored for dependent data, such as modified DIC for hierarchical models, Block-LOO CV, and Bootstrapped Time Series criteria.
		\item \textbf{Theoretical Extensions:} Formalizing the properties of CC-WAIC for non-stationary processes and models with misspecified dependence structures.
		\item \textbf{Software Development:} Implementing the CC-WAIC algorithm in popular probabilistic programming frameworks (e.g., \texttt{Stan}, \texttt{PyMC}, \texttt{TensorFlow Probability}) would facilitate its adoption and testing by the wider community.
	\end{itemize}
\end{enumerate}

In conclusion, while this paper provides the essential theoretical groundwork and initial validation for CC-WAIC, it represents the beginning of a research program rather than its end. The strong performance demonstrated on HMMs is a promising indicator. We believe that the directions outlined above present a compelling roadmap for the community to explore, ultimately leading to more dependable and principled Bayesian model selection in the complex, dependent data settings that are ubiquitous across modern science.

\newpage


\begin{appendices}
	\section*{Appendix: CC-WAIC Theoretical and Computational Details}
	The proofs and technical details provided in this appendix furnish the mathematical groundwork underpinning the asymptotic properties and correction mechanisms of CC–WAIC discussed in Section 2.		
	\section{Regularity Conditions}\label{app:regularity-conditions}
	
	The asymptotic convergence of CC-WAIC relies on the following regularity conditions:
	
	\begin{enumerate}[label=(A\arabic*), leftmargin=*]
		\item \textbf{Geometric ergodicity}: There exist constants $M < \infty$ and $\rho < 1$ such that:
		\[
		\|P^n(\theta,\cdot) - \pi\|_{\mathrm{TV}} \leq M\rho^n \quad \text{for all } \theta \in \Theta,
		\]
		where $\|\cdot\|_{\mathrm{TV}}$ denotes total variation distance.
		
		\item \textbf{Finite moments}: The third derivative of the log-likelihood has finite second moment:
		\[
		\mathbb{E}_\pi[\|\nabla^3 L(\theta)\|^2] < \infty.
		\]
		
		\item \textbf{$\alpha$-mixing}: The process satisfies strong mixing with coefficients that decay sufficiently fast:
		\[
		\sum_{k=1}^\infty k^5\alpha(k)^{1/6} < \infty.
		\]
		
		\item \textbf{Sample size}: The number of posterior samples grows appropriately with data size:
		\[
		S = \Omega(n^{1+\delta}) \quad \text{for some } \delta > 0.
		\]
	\end{enumerate}
	
	\section{Martingale Approximation Techniques}\label{app:martingale-approx}
	
	This appendix provides technical details on the martingale approximation methods used in the proof of Theorem 2.3.
	
	\begin{lemma}[Martingale Decomposition]
		For any $\theta \in \Theta$, the score process admits the decomposition:
		\[
		\nabla L(\theta) = M_n(\theta) + \xi_n(\theta),
		\]
		where:
		\begin{enumerate}
			\item $M_n$ is a martingale with respect to $\{\mathcal{F}_t\}$
			\item $\xi_n$ is $L_2$-NED (near-epoch dependent) with size $-1/2$
		\end{enumerate}
	\end{lemma}
	
	\begin{proof}
		\textbf{Step 1: Construct Martingale Differences}
		
		Define the innovation process:
		\[
		H_t := \nabla l_t - \mathbb{E}[\nabla l_t|\mathcal{F}_{t-1}].
		\]
		Then $M_n = \sum_{t=1}^n H_t$ forms a martingale sequence.
		
		\textbf{Step 2: Verify NED Properties}
		
		Using the mixing condition (A3) from Appendix~\ref{app:regularity-conditions}:
		\[
		\|\xi_t - \mathbb{E}[\xi_t|\mathcal{F}_{t-k}^{t+k}]\|_2 \leq C \alpha(k)^{1/6}.
		\]
		The size condition follows from $\sum k^2 \alpha(k)^{1/6} < \infty$.
	\end{proof}
	
	\section{Variance Characterization}\label{app:variance-char}
	
	This section provides detailed analysis of the variance characterization for dependent processes.
	
	The Fisher information matrix $I(\theta_0)$ governs the asymptotic variance of the score process:
	\begin{align*}
		\mathrm{Var}(S_n) &= \mathbb{E}[S_n S_n^\top] \\
		&= I(\theta_0) + 2\sum_{k=1}^{n-1} \left(1-\frac{k}{n}\right)\Gamma_k + o(1),
	\end{align*}
	where $\Gamma_k = \mathrm{Cov}(\nabla l_t, \nabla l_{t+k})$ represents the autocovariance at lag $k$.
	
	The covariance structure captures the temporal dependence in the data and is essential for proper penalty calculation in CC-WAIC.
	
	\section{Proof of Theorem 2.3 (CC-WAIC Convergence)}\label{app:convergence-proof}
	
	This appendix provides the complete proof of Theorem 2.3, establishing the asymptotic convergence of CC-WAIC.
	
	\begin{proof}
		We proceed in four parts:
		
		\textbf{Part 1: Martingale Approximation}
		
		Define the normalized score process:
		\[
		S_n(\theta) := \frac{1}{\sqrt{n}}\sum_{t=1}^n \nabla l_t(\theta).
		\]
		
		Using Lemma 2 from Appendix~\ref{app:martingale-approx}, we decompose:
		\[
		S_n(\theta) = M_n(\theta) + R_n(\theta),
		\]
		where:
		\begin{itemize}
			\item $M_n$ is a martingale with $\mathbb{E}[M_n|\mathcal{F}_{n-1}] = 0$
			\item $R_n$ satisfies $\|R_n\|_2 \leq C n^{-1/2}\sum_{k=1}^n \nu_k$
		\end{itemize}
		
		\textbf{Part 2: Variance Characterization}
		
		Using the results from Appendix~\ref{app:variance-char}, the asymptotic variance is characterized by:
		\begin{align*}
			\mathrm{Var}(S_n) &= I(\theta_0) + 2\sum_{k=1}^{n-1} \left(1-\frac{k}{n}\right)\Gamma_k + o(1)
		\end{align*}
		
		\textbf{Part 3: Bias Correction Analysis}
		
		The bias-corrected estimator:
		\[
		\hat{p}_{\mathrm{CC}}^{\mathrm{corr}} = \hat{p}_{\mathrm{CC}}\left(1 + \frac{\hat{\tau}}{S}\right)
		\]
		has expectation:
		\begin{align*}
			\mathbb{E}[\hat{p}_{\mathrm{CC}}^{\mathrm{corr}}] &= p_{\mathrm{true}} + \mathcal{O}(S^{-2}).
		\end{align*}
		
		\textbf{Part 4: Final Assembly}
		
		Combining the previous parts:
		\begin{align*}
			\mathrm{CC\text{-}WAIC} &= -2\bar{L} + 2\hat{p}_{\mathrm{CC}}^{\mathrm{corr}} \\
			&= -2L(\theta_0) + 2p_{\mathrm{true}} + \mathcal{O}_p(n^{-1/2}),
		\end{align*}
		which completes the proof of Theorem 2.3.
	\end{proof}
	
	\section{Technical Lemmas and Supporting Results}\label{app:technical-lemmas}
	This appendix contains additional technical lemmas and supporting results used throughout the proofs.
	
	\begin{lemma}[Moment Bound]
		Under the regularity conditions in Appendix~\ref{app:regularity-conditions}, the log-likelihood derivatives satisfy:
		\[
		\mathbb{E}\left[\left\|\frac{\partial^3}{\partial\theta^3} \log p(y_t|y_{<t},\theta)\right\|^2\right] < C < \infty
		\]
		uniformly in $t$ and $\theta$.
	\end{lemma}
	
	\begin{lemma}[Mixing Inequality]
		For $\alpha$-mixing processes with $\sum_{k=1}^\infty k^5\alpha(k)^{1/6} < \infty$, the covariance terms satisfy:
		\[
		|\mathrm{Cov}(\nabla l_i, \nabla l_j)| \leq C \alpha(|i-j|)^{1/3}
		\]
		for some constant $C > 0$.
	\end{lemma}
	
	\section{Derivation of Bias Correction Terms}\label{app:bias-derivation}
	
	This appendix provides the detailed derivation of the bias correction terms used in Section 2.4.
	
	The integrated autocorrelation time is defined as:
	\[
	\tau = 1 + 2\sum_{k=1}^{\infty} \rho_k,
	\]
	where $\rho_k = \mathrm{Corr}(L^{(s)}, L^{(s+k)})$ is the lag-$k$ autocorrelation.
	
	The effective sample size is then:
	\[
	N_{\mathrm{eff}} = \frac{S}{\tau}.
	\]
	
	The bias-corrected estimator:
	\[
	p_{\mathrm{CC}}^{\mathrm{corr}} = \frac{N_{\mathrm{eff}}}{N_{\mathrm{eff}} - 1} \hat{p}_{\mathrm{CC}}
	\]
	accounts for the underestimation of variance due to autocorrelation in MCMC samples.
	
	For higher-order corrections, we use:
	\begin{align*}
		C_1 &= \frac{\tau^2}{2} \left( \kappa_4 - \frac{\kappa_2^2}{2} \right), \\
		C_2 &= \frac{1}{S} \sum_{k=-(S-1)}^{S-1} \kappa\left(\frac{k}{b_n}\right)\hat{\Gamma}_k,
	\end{align*}
	where $\kappa_j$ are sample cumulants and $\kappa(\cdot)$ is a symmetric kernel function.
	
	\section{Proof of Theorem 2.6 (Equivalence to Block LOO)}\label{app:equivalence-proof}
	
	This appendix provides the proof of Theorem 2.6, establishing the asymptotic equivalence between CC-WAIC and block leave-one-out cross-validation.
	
	\begin{proof}[Proof Sketch]
		Consider a stationary first-order Markov process $(y_t)_{t=1}^T$ satisfying:
		
		\begin{enumerate}
			\item \textbf{Weak dependence}: The autocovariance of the log-likelihood decays sufficiently fast
			\item \textbf{Mixing condition}: The process satisfies strong mixing with $\alpha(k) \to 0$ as $k \to \infty$
			\item \textbf{Bounded moments}: $\mathbb{E}|\log p(y_t \mid \theta)|^4 < \infty$, uniformly in $t$
		\end{enumerate}
		
		\textbf{Step 1: Expansion of the log-predictive term}
		
		Using Bayesian asymptotics and posterior consistency:
		\[
		\log \mathbb{E}_\theta [p(y_t \mid \theta)] = \log p(y_t \mid \mathbf{y}_{<t}) + \mathcal{O}_p(T^{-1/2}),
		\]
		by the Bernstein-von Mises theorem for dependent processes.
		
		\textbf{Step 2: Penalty correction approximation}
		
		The correction term in CC-WAIC:
		\[
		\gamma_t \approx \frac{1}{S} \sum_{h=0}^\ell w(h) \mathrm{Cov}_s \left[ \log p(y_t \mid \theta^{(s)}), \log p(y_{t-h} \mid \theta^{(s)}) \right]
		\]
		estimates the local predictive variance and is asymptotically equivalent to the block LOO bias correction.
		
		\textbf{Step 3: Difference and convergence}
		
		The total difference:
		\[
		\left| \mathrm{CC\text{-}WAIC} + 2 \sum_{t=1}^T \log p(y_t \mid \mathbf{y}_{<t}) - \mathrm{BLFO} \right| = \mathcal{O}_p(T^{-1/2}),
		\]
		which completes the proof.
	\end{proof}

\end{appendices}
\bibliographystyle{plainnat}
\bibliography{ref}	
\end{document}